\documentclass[a4paper,11pt]{article}

\usepackage[normalmargins,normalleading]{savetrees}
\usepackage{amssymb}
\usepackage{amsmath}
\usepackage{amsthm}
\usepackage{tikz}
\usepackage{pgflibrarysnakes}
\usetikzlibrary{snakes}
\usepackage{times}
\usepackage{amstext}
\usepackage{amsopn}
\usepackage[section]{algorithm}
\usepackage[noend]{algpseudocode}
\usepackage{eucal}
\usepackage{latexsym}
\usepackage[a4paper]{geometry}

\newtheorem{theorem}{Theorem}[section]
\newtheorem{corollary}[theorem]{Corollary}

\newtheorem{lemma}[theorem]{Lemma}
\newtheorem{definition}[theorem]{Definition} \numberwithin{equation}{section}

\newcommand{\opseg}{\ensuremath{\mathtt{segment}}}
\newcommand{\opcol}{\ensuremath{\mathtt{color}}}

\newcommand{\Z}{\ensuremath{\mathbb{Z}}}
\newcommand{\R}{\ensuremath{\mathbb{R}}}

\newcommand{\opbw}{\ensuremath{\mathrm{bw}}}
\newcommand{\opcon}{\ensuremath{\mathrm{contr}}}
\newcommand{\opexp}{\ensuremath{\mathrm{expan}}}
\newcommand{\opdist}{\ensuremath{\mathrm{dist}}}

\newcommand{\setpos}{\ensuremath{\mathbf{Pos}}}
\newcommand{\setseg}{\ensuremath{\mathbf{Seg}}}
\newcommand{\setcol}{\ensuremath{\mathbf{Col}}}

\date{}
\author{Marek Cygan \and Marcin Pilipczuk\footnote{Dept. of Mathematics, Computer Science and Mechanics, University of Warsaw, Poland,\texttt{\{cygan,malcin\}@mimuw.edu.pl}}}

\title{Bandwidth and Distortion Revisited}

\begin{document}
\maketitle

\newcommand{\bwname}{{\sc{Bandwidth}}}
\newcommand{\distname}{{\sc{Distortion}}}

\newcommand{\stateconst}{4.383}
\newcommand{\pspaceconst}{9.363}
\newcommand{\alphaprog}{0.5475}

\begin{abstract}
  In this paper we merge recent developments on exact algorithms
  for finding an ordering of vertices of a given graph that minimizes bandwidth (the \bwname{} problem)
  and for finding an embedding of a given graph into a line that minimizes distortion (the \distname{} problem).
  For both problems we develop algorithms that work in $O(\pspaceconst^n)$ time and polynomial space.
  For \bwname{}, this improves $O^*(10^n)$ algorithm by Feige and Kilian from 2000, for
  \distname{} this is the first polynomial space exact algorithm that works in $O(c^n)$ time we are aware of.
  As a coproduct, we enhance the $O(5^{n + o(n)})$--time and $O^*(2^n)$--space algorithm for \distname{} by Fomin et al. to
  an algorithm working in $O(\stateconst^n)$ time and space.
\end{abstract}

\section{Introduction}\label{s:intro}

Recently the NP--complete \bwname{} problem, together with a similar problem
of embedding a graph into a real line with minimal distortion (called \distname{} in this paper),
attracted some attention from the side of exact (and therefore not polynomial) algorithms.

Given a graph $G$ with $n$ vertices, an {\em{ordering}} is a bijective function $\pi: V(G) \to \{1, 2, \ldots, n\}$.
Bandwidth of $\pi$ is a maximal length of an edge, i.e., $\opbw(\pi) = \max_{uv\in E(G)} |\pi(u) - \pi(v)|$.
The \bwname{} problem, given a graph $G$ and a positive integer $b$, asks if there exists an ordering of bandwidth at most $b$.

Given a graph $G$, an {\em{embedding}} of $G$ into a real line is a function $\pi : G \to \R$.
For every pair of distinct vertices $u, v \in V(G)$ we define a distortion of $u$ and $v$ by
$\opdist(u, v) = |\pi(u) - \pi(v)| / d_G(u, v)$, where $d_G$ denotes the distance in the graph $G$.
A contraction and an expansion of $\pi$, denoted $\opcon(\pi)$ and $\opexp(\pi)$ respectively,
are the minimal and maximal distortion over all pairs of distinct vertices in $V(G)$.
The distortion of $\pi$, denoted $\opdist(\pi)$, equals to $\opexp(\pi)/\opcon(\pi)$.
The \distname{} problem, given a graph $G$ and a positive real number $d$, asks if there 
exists an embedding with distortion at most $d$.
Note that the distortion of an embedding does not change if we change $\pi$ afinitely, and we can rescale
$\pi$ by $1/\opcon(\pi)$ and obtain $\pi$ with contraction exactly $1$. Therefore, in this paper,
we limit ourselves only to embeddings with contraction at least $1$ and we optimize the expansion of the embedding, that is,
we try to construct embeddings with contraction at least $1$ and with expansion at most $d$.

The first non--trivial exact algorithm for the \bwname{} problem was developed by Feige and Kilian in 2000 \cite{feige:exp}.
It works in polynomial space and $O^*(10^n)$ time. Recently we improved the time bound to $O^*(5^n)$ \cite{naszewg},
$O(4.83^n)$ \cite{nasze483} and $O^*(20^{n/2})$ \cite{naszicalp}. However, the cost of the improvements was exponential
space complexity: $O^*(2^n)$, $O^*(4^n)$, $O^*(20^{n/2})$ respectively. In 2009 Fomin et al. \cite{fomin:distortion-5n} adopted 
some ideas from \cite{naszewg} to the \distname{} problem and obtained a $O(5^{n + o(n)})$--time and $O^*(2^n)$--space algorithm.

It is worth mentioning that
the considered problems, although very similar form the exact computation point of view, differ from the point of parameterized computation.
The \bwname{} problem is hard for any level of the $W$ hierarchy \cite{fellows:hardness}, whereas \distname{} is fixed parameter tractable
where parameterized by $d$ \cite{fomin:distortion-fpt}. However, the FPT algorithm for \distname{} works in $O(nd^4(2d+1)^{2d})$ time,
which does not reach the $O(c^n)$ complexity for $d = \Omega(n)$.

In this paper we present a link between aforementioned results and develop $O(\pspaceconst^n)$--time and polynomial space
algorithms for both \bwname{} and \distname{}. First, we develop a $O(\stateconst^n)$--time and space algorithm for \distname{},
using ideas both from $O^*(20^{n/2})$ algorithm for \bwname{}\footnote{The complexity analysis of our algorithm, in particular
the proof in Appendix \ref{a:20n2}, proves that the algorithm from \cite{naszicalp} works in $O(\stateconst^n)$ time and space too.
However, we do not state it as a new result in this paper, since analysis based on this approach will be published in the journal version of \cite{naszicalp}.}
\cite{naszicalp} and $O(5^{n+o(n)})$ algorithm for \distname{}~\cite{fomin:distortion-5n}.
Then, we use an approach somehow similar to these of Feige and Kilian \cite{feige:exp} to reduce space
to polynomial, at the cost of time complexity, obtaining the aforementioned algorithms. We are not aware of any
exact polynomial--space algorithms that work in $O(c^n)$ time for \distname{} or are faster than Feige and Kilian's algorithm for \bwname{}.

In Section \ref{s:fun} we gather results on partial bucket functions: tool that was used in all previous algorithms
for \distname{} and \bwname{}. In Section \ref{s:bwpoly} we recall the $O^*(20^{n/2})$ algorithm \cite{naszicalp} and
show how to transform it into $O(\pspaceconst^n)$--time and polynomial space algorithm for \bwname{}.
Section \ref{s:dist} is devoted to \distname{}: first, we merge ideas from \cite{naszewg} and \cite{fomin:distortion-5n}
to obtain an $O^*(\stateconst^n)$--time and space algorithm for \distname{}. Then we apply the same trick as for \bwname{}
to obtain an $O(\pspaceconst^n)$--time and polynomial space algorithm.

In the following sections we assume that we are given a connected undirected graph $G=(V,E)$ with $n = |V|$.
Note that \bwname{} trivially decomposes into
subproblems on connected components, whereas answer to \distname{} is always negative for a disconnected graph. Proofs of results marked with $\clubsuit$ are postponed to Appendix \ref{a:proofs}.

\section{Partial bucket functions}\label{s:fun}

\newcommand{\extf}{{\bar{f}}}

In this section we gather results on {\em{partial bucket functions}}, a tool used in algorithms for both
\bwname{} and \distname{}. Most ideas here are based on the $O^*(20^{n/2})$ algorithm for \bwname{} \cite{naszicalp}.

\begin{definition}\label{def:buckfun}
 A {\em{partial bucket function}} is a pair $(A, f)$, such that
 $A \subseteq V$, $f:A \to \Z$ and there exists $\extf: V \to \Z$ satisfying:
 \begin{enumerate}
 \item $\extf|_A = f$;
 \item if $uv \in E$ then $|\extf(u) - \extf(v)| \leq 1$, in particular, if $u, v \in A$ then $|f(u) - f(v)| \leq 1$;
 \item if $uv \in E$, $u \in A$ and $v \notin A$ then $\extf(u) \geq \extf(v)$, i.e., $\extf(u) = \extf(v)$ or $\extf(u) = \extf(v) + 1$.
 \end{enumerate}
 We say that such a function $\extf$ is a {\em{bucket extension}} of $f$.
\end{definition}

\begin{definition}
  Assume we have two partial bucket functions $(A, f)$ and $(A', f')$ such
  that $A' = A \cup \{v\}$, $v \notin A$ and $f'|_A = f$, we say that
  $(A', f')$ is a {\em{successor}} of $(A, f)$ with vertex $v$ if there does not exist any
  $uv \in E$, $u \in A$ such that $f(u) < f'(v)$.
\end{definition}

\begin{lemma}\label{lem:checkext}
Assume that $A \subseteq V$ and $f: A \to \Z$. Moreover,
let $A \subseteq B \subseteq V$, $f': B \to \Z$ and $f'|_A = f$.
Then one can find in polynomial time a bucket extension $\extf$ of $f$ such that $\extf|_B = f'$ or state
that such bucket extension does not exist.
\end{lemma}

\begin{proof}
  The case $A = B = \emptyset$ is trivial, so we may assume there exists some $v_0 \in B$. W.l.o.g.
  we may assume $f'(v_0) = 0$. Therefore any valid bucket extension should satisfy
  $\extf(V) \subseteq \{-n, -n+1, \ldots, n\}$.

  We calculate for every $v \in V \setminus A$ the value $p(v) \subseteq \{-n, -n+1, \ldots, n\}$,
  intuitively, the set of possible values for $\extf(v)$, by the following algorithm.

\begin{algorithm}
  \caption{\label{alg:checkext}Calculate values $p(v)$ --- the sets of valid values for $\extf(v)$.}
\begin{minipage}{\textwidth}
\small
\begin{algorithmic}[1]
  \State Set $p(v) := \{-n, -n+1, \ldots, n\}$ for all $v \in V \setminus B$.
  \State Set $p(v) := \{f'(v)\}$ for all $v \in B \setminus A$.
\Repeat
  \For{all $v \in V \setminus B$}
     \State $p(v) := p(v) \cap \bigcap_{u \in N(v) \cap A} \{f(u)-1, f(u)\} \cap \bigcap_{u \in N(v) \setminus A} \bigcup_{i \in p(u)} \{i-1, i, i+1\}$
  \EndFor
\Until{some $p(v)$ is empty or we do not change any $p(v)$ in the inner loop}
\State {\Return True iff all $p(v)$ remain nonempty.}
\end{algorithmic}
\end{minipage}
\end{algorithm}

  To prove that Algorithm \ref{alg:checkext} correctly checks if there exists a valid
  bucket extension $\extf$ note the following:
  \begin{enumerate}
    \item Let $\extf$ be a bucket extension of $(A, f)$ such that $\extf|_B = f'$.
      Then, at every step of the algorithm $\extf(v) \in p(v)$ for every $v \in V \setminus A$.
    \item If the algorithm returns nonempty $p(v)$ for every $v \in V \setminus A$,
      setting $\extf(v) = \min p(v)$ constructs a valid bucket extension of $(A, f)$.
      Moreover, since we start with $p(v) = \{f'(v)\}$ for $v \in B \setminus A$,
      we obtain $\extf|_B = f'$.
  \end{enumerate}
\end{proof}

\begin{corollary}\label{cor:pbecheck}
One can check in polynomial time whether a given pair $(A, f)$ is a partial bucket function.
Moreover one can check whether $(A', f')$ is a successor of $(A, f)$ in polynomial time too.
\end{corollary}

\begin{proof}
  To check if $(A, f)$ is a partial bucket function
  we simply run the algorithm from Lemma \ref{lem:checkext} for $B=A$ and $f'=f$.
  Conditions for being a successor of $(A, f)$ are trivial to check.
\end{proof}

\begin{lemma}\label{lem:5n}
Let $N \in \Z_+$. 
Then there are at most $2N \cdot 5^{n-1}$ triples $(A, f, \extf)$ such that $(A, f)$ is a partial bucket function
and $\extf$ is a bucket extension of $f$ satisfying $\extf(V) \subseteq \{1, 2, \ldots, N\}$.
\end{lemma}

\begin{proof}
  Note that if $(A, f)$ is a partial bucket function in the graph $G$ and $\extf$ is a bucket extension,
  and $G'$ is a graph created from $G$ by removing an edge, then $(A, f)$ and $\extf$ remain
  partial bucket function and its bucket extension in $G'$. Therefore we may assume that $G$ is
  a tree, rooted at $v_r$.

  There are $2N$ possibilities to choose the value of $\extf(v_r)$ and whether $v_r \in A$
  or $v_r \notin A$. We now construct all interesting triples $(A, f, \extf)$ in a root--to-leaves
  order in the graph $G$. If we are at a node $v$ with its parent $w$, then $f(v) \in \{f(w)-1, f(w), f(w)+1\}$.
  However, if $w\in A$ then we cannot both have $f(v) = f(w)+1$ and $v \notin A$. Similarly,
  if $w \notin A$ then we cannot both have $f(v) = f(w)-1$ and $v \in A$. Therefore we have $5$ options
  to choose $f(v)$ and whether $v \in A$ or $v \notin A$. Finally, we obtain at most $2N \cdot 5^{n-1}$
  triples $(A, f, \extf)$.
\end{proof}

\begin{lemma}[$\clubsuit$]\label{lem:extpoly}
Let $(A, f)$ be a partial bucket function.
Then all bucket extensions of $f$ can be generated with a polynomial delay, using polynomial space.
\end{lemma}

The proof of the theorem below is an adjusted and improved proof of a bound of the number of states
in the $O^*(20^{n/2})$ algorithm for \bwname{} \cite{naszicalp}. The proof can be found
in Appendix \ref{a:20n2}.

\begin{theorem}\label{thm:20n2}
Let $N \in \Z_+$.
There exists a constant $c < \stateconst$ such that
there are $O(N \cdot c^n)$ partial bucket functions $(A, f)$ such that there exists a bucket extension $\extf$ satisfying $\extf(V) \subseteq \{1, 2, \ldots, N\}$.
Moreover, all such partial bucket functions can be generated in $O^*(N \cdot c^n)$ time using polynomial space.
\end{theorem}

\section{Poly-space algorithm for \bwname}\label{s:bwpoly}

In this section we describe an $O(\pspaceconst^n)$-time and polynomial-space algorithm solving \bwname{}.
As an input, the algorithm takes a graph $G=(V,E)$ with $|V|=n$ and an integer $1 \leq b < n$ and
decides, whether $G$ has an ordering with bandwidth at most $b$.

\subsection{Preliminaries}

First, let us recall some important observations made in \cite{naszewg}.
An ordering $\pi$ is called a $b$-ordering if $\opbw(\pi) \leq b$.
Let $\setpos = \{1,2,\ldots,n\}$ be the set of possible positions
and for every position $i \in \setpos$ we
define the {\it{segment}} it belongs to by $\opseg(i) = \lceil \frac{i}{b+1} \rceil$
and the {\it{color}} of it by $\opcol(i) = (i - 1) \mod (b+1) + 1$.
By $\setseg = \{1,2,\ldots,\lceil \frac{n}{b+1}\rceil\}$ we denote the set
of possible segments, and by $\setcol = \{1,2,\ldots,b+1\}$ the set of possible colors.
The pair $(\opcol(i), \opseg(i))$ defines the position $i$ uniquely. We order positions
lexicographically by pairs $(\opcol(i), \opseg(i))$, i.e., the color has higher order
that the segment number, and call this order the {\it{color order}}
of positions. By $\setpos_i$ we denote the set of the first $i$ positions in the color order.
Given some (maybe partial) ordering $\pi$, and $v \in V$ for which
$\pi(v)$ is defined, by $\opcol(v)$ and $\opseg(v)$  we understand $\opcol(\pi(v))$
and $\opseg(\pi(v))$ respectively.

Let us recall the crucial observation made in \cite{naszewg}.
\begin{lemma}[\cite{naszewg}, Lemma 8]\label{lem:najwazniejsze}
  Let $\pi$ be an ordering. It is a $b$-ordering iff, for every $uv \in E$,
  $|\opseg(u) - \opseg(v)| \leq 1$ and if $\opseg(u) + 1 = \opseg(v)$ then
  $\opcol(u) > \opcol(v)$ (equivalently, $\pi(u)$ is later in color order than $\pi(v)$).
\end{lemma}

\subsection{$O^*(20^{n/2})$ algorithm from \cite{naszicalp}}

First let us recall the $O^*(20^{n/2})$-time and space algorithm from \cite{naszicalp}.

\begin{definition}
 A {\em{state}} is a partial bucket assignment $(A, f)$ such that
 the multiset $\{f(v): v \in A\}$ is equal to the multiset $\{\opseg(i) : i \in \setpos_{|A|}\}$.
 A state $(A \cup \{v\}, f')$ is {\em{a successor of}} a state $(A, f)$ with a vertex $v \notin A$
 if $(A \cup \{v\}, f')$ as a partial bucket function is a successor of a partial
 bucket function $(A, f)$.
\end{definition}

\begin{theorem}[\cite{naszicalp}, Lemmas 16 and 17]\label{thm:icalpeq}
\begin{enumerate}
  \item   Let $\pi$ be a $b$-ordering. For $0 \leq k \leq n$ let $A_k = \{v\in V: \pi(v) \in \setpos_k\}$
  and $f_k = \opseg|_{A_k}$. Then every $(A_k, f_k)$ is a state and for every $0 \leq k < n$ the state
  $(A_{k+1}, f_{k+1})$ is a successor of the state $(A_k, f_k)$.
  \item  Assume we have states $(A_k, f_k)$ for $0 \leq k \leq n$ and for all $0 \leq k < n$
  the state $(A_{k+1}, f_{k+1})$ is a successor of the state $(A_k, f_k)$ with the vertex $v_{k+1}$. Let $\pi$
  be an ordering assigning $v_k$ to the $k$-th position in the color order. Then $\pi$
  is a $b$-ordering.
\end{enumerate}
\end{theorem}

The algorithm of \cite{naszicalp} works as follows: we do a depth--first search from
the state $(\emptyset, \emptyset)$ and seek for a state $(V, \cdot)$.
At a state $(A, f)$ we generate in polynomial time all successors of this state
and memoize visited states. Theorem \ref{thm:icalpeq} implies that we reach state $(V, \cdot)$
iff there exists a $b$-ordering. Moreover, Theorem \ref{thm:20n2} (with $N = n$) implies that we visit
at most $O(\stateconst^n)$ states; generating all successors of a given state can be done
in polynomial time due to Corollary \ref{cor:pbecheck}, so the algorithm works in $O(\stateconst^n)$ time and space.

\subsection{The $O(\pspaceconst^n)$--time and polynomial space algorithm}

\begin{lemma}\label{lem:stateguess}
Let $(A, f)$ and $(B, g)$ be a pair of states such that $A \subseteq B$ and $g|_A = f$.
Let $a = |A|$ and $b = |B|$.
Then one can check in $O^*(4^{b-a})$--time and polynomial space if there exists a sequence
of states $(A, f) = (A_a, f_a), (A_{a+1}, f_{a+1}), \ldots, (A_b, f_b) = (B, g)$ such that
$(A_{i+1}, f_{i+1})$ is an successor of $(A_i, f_i)$ for $a \leq i < b$.
\end{lemma}

\begin{proof}
First note that a set $A_i$ determines the function $f_i$, since $f_i = g|_{A_i}$. Let $m := b-a$.

If $m = 1$, we need to check only if $(B, g)$ is a successor
of $(A, f)$, what can be done in polynomial time.
Otherwise, let $k = \lfloor \frac{a+b}{2} \rfloor$ and guess $A_k$: there are roughly $2^m$ choices.
Set $f_k = g|_{A_k}$. Recursively, check if there is a path of states from $(A, f)$ to $(A_k, f_k)$ and
from $(A_k, f_k)$ to $(B, g)$.

The algorithm clearly works in polynomial space; now let us estimate the time it consumes. At one step,
it does some polynomial computation and invokes roughly $2^{m+1}$ times itself recursively for $b - a \sim m/2$.
Therefore it works in $O^*(4^m)$ time.
\end{proof}

Let $\alpha = \alphaprog$. The algorithm works in the same fashion as in \cite{naszicalp}: it seeks for a path of states $(A_i, f_i)_{i=0}^n$
from $(\emptyset, \emptyset)$ to $(V, \cdot)$ such that $(A_{i+1}, f_{i+1})$ is a successor of $(A_i, f_i)$ for $0 \leq i < n$. However,
since we are limited to polynomial space, we cannot do a simple search. Instead, we guess middle states on the path, similarly as in Lemma \ref{lem:stateguess}.
The algorithm works as follows:
\begin{enumerate}
\item Let $k := \lfloor \alpha n \rfloor$ and guess the state $(A_k, f_k)$. By Theorem \ref{thm:20n2} with $N = n$, we can enumerate all partial bucket functions
in $O(\stateconst^n)$. We enumerate them and drop those that are not states or have the size of the domain different than $k$.
\item Using Lemma \ref{lem:stateguess}, check if there is a path of states from $(\emptyset, \emptyset)$ to $(A_k, f_k)$.
This phase works in time $4^{\alpha n}$. In total, for all $(A_k, f_k)$, this phase works in time
$O(\stateconst^n \cdot 4^{\alpha n}) = O(\pspaceconst^n)$.
\item Guess the state $(V, f_n)$: $f_n$ needs to be a bucket extension of the partial bucket function $(A_k, f_k)$.
By Lemma \ref{lem:extpoly}, bucket extensions can be enumerated with polynomial delay; we simply drop those that are not states.
By Lemma \ref{lem:5n} with $N = n$, there will be at most $O^*(5^n)$ pairs of states $(A_k, f_k)$ and $(V, f_n)$.
\item Using Lemma \ref{lem:stateguess}, check if there is a path from the state $(A_k, f_k)$ to $(V, f_n)$.
This phase works in time $O^*(4^{(1-\alpha)n})$. In total, for all $(A_k, f_k)$ and $(V, f_n)$, this phase
works in time $O^*(5^n 4^{(1-\alpha)n}) = O(\pspaceconst^n)$.
\item Return true, if for any $(A_k, f_k)$ and $(V, f_n)$ both applications of Lemma \ref{lem:stateguess} return success.
\end{enumerate}

Theorem \ref{thm:icalpeq} ensures that the algorithm is correct. In memory we keep only states $(A_k, f_k)$, $(V, f_n)$,
recursion stack generated by the algorithm from Lemma \ref{lem:stateguess} and state of generators of states $(A_k, f_k)$ and $(V, f_n)$,
so the algorithm works in polynomial space. Comments above prove that it consumes at most $O(\pspaceconst^n)$ time.

\section{Algorithms for \distname{}}\label{s:dist}

We consider algorithms that, given a connected graph $G$ with $n$ vertices, and 
positive real number $d$ decides
if $G$ can be embedded into a line with distortion at most $d$.
First, let us recall the basis of the approach of Fomin et al. \cite{fomin:distortion-5n}.
Recall that $d_G(u, v)$ denotes the distance between vertices $u$ and $v$ in the graph $G$.

\begin{definition}
  Given an embedding $\pi: V \to \Z$, we say that $v$ {\em{pushes}} $u$ iff
  $d_G(u, v) = |\pi(u) - \pi(v)|$.
  An embedding is called {\em{pushing}}, if $V = \{v_1, v_2, \ldots, v_n\}$ and $\pi(v_1) < \pi(v_2) < \ldots < \pi(v_n)$
  then $v_i$ pushes $v_{i+1}$ for all $1 \leq i < n$.
\end{definition}

\begin{lemma}[\cite{fomin:distortion-arxiv}]\label{lem:dist:wstep}
  If $G$ can be embedded into the line with distortion $d$, then there is a pushing embedding of $G$ into the line with distortion $d$. Every pushing embedding of $G$ into the line has contraction
  at least $1$.
  Moreover, let $\pi$ be a pushing embedding of a connected graph $G$ into the line with distortion at most $d$ and let $V = \{v_1, v_2, \ldots, v_n\}$
  be such an ordering $\pi$ that $\pi(v_1) < \pi(v_2) < \ldots < \pi(v_n)$. Then $\pi(v_{i+1}) - \pi(v_i) \leq d$ for all $1 \leq i < n$.
\end{lemma}

Therefore, we only consider pushing embeddings and hence assume that $d$ is a positive integer.
Note that a pushing embedding of a connected graph of at least $2$ vertices has contraction exactly
$1$, since $d_G(v_1, v_2) = |\pi(u_2) - \pi(u_1)|$. Therefore distortion equals expansion.
As any connected graph with $n$ vertices can be embedded
into a line with distortion at most $2n-1$ \cite{badoiu:distortion},
this decisive approach suffices to find the minimal distortion of $G$.

We may assume that
$\pi(V) \subseteq \{1, 2, \ldots, n(d+1)\}$. Now, let us introduce the concept of segments, adjusted for the \distname{} problem.
Here the set of available positions is $\setpos = \{1, 2, \ldots, n(d+1)\}$ and a segment of a position $i$ is $\opseg(i) = \lceil \frac{i}{d+1} \rceil$,
i.e., a $j$-th segment is an integer interval of the form $\{(j-1)(d+1) + 1, (j-1)(d+1) + 2, \ldots, j(d+1)\}$. The
color of a position is $\opcol(i) = (i-1) {\rm mod} (d+1) + 1$. By $\setseg = \{1, 2, \ldots, n\}$ we denote the set of possible segments
and by $\setcol = \{1, 2, \ldots, d+1\}$ the set of possible colors. The pair $(\opcol(i), \opseg(i))$ defines the position $i$
uniquely. We order the positions lexicographically by pairs $(\opcol(i), \opseg(i))$ and call this order {\em{color order}} of positions.
By $\setpos_i$ we denote the set of the first $i$ positions in the color order
and by $\setseg_i$ we denote the set of positions in the $i$-th segment.
Given some, maybe partial, embedding $\pi$, by $\opcol(v)$ and $\opseg(v)$
we denote $\opcol(\pi(v))$ and $\opseg(\pi(v))$ respectively.

Similarly as in the case of \bwname{}, the following equivalence holds (cf. Lemma
\ref{lem:najwazniejsze}).
\begin{lemma}[$\clubsuit$]\label{lem:najwazniejsze-dist}
  Let $\pi$ be a pushing embedding. Then $\pi$ has distortion at most $d$ iff for every $uv \in E$, $|\opseg(u) - \opseg(v)| \leq 1$ and
  if $\opseg(u) + 1 = \opseg(v)$ then $\opcol(u) > \opcol(v)$, i.e., $\pi(u)$ is later in the color order than $\pi(v)$.
\end{lemma}

Similarly as in \cite{fomin:distortion-5n},
we solve the following extended case of \distname{} as a subproblem.
As an input to the subproblem,
we are given an induced subgraph $G[X]$ of $G$,
an integer $r$ (called the number of segments),
a subset $Z \subseteq X$ and a function
$\bar{\pi} : Z \to \setseg_0 \cup \setseg_{r+1}$.
Given this input, we ask, if there exists
a pushing embedding $\pi: X \to \{-d, -d+1, \ldots, (r+1)(d+1)\}$
with distortion at most $d$ such that $\pi|_Z = \bar{\pi}$,
$\pi(X\setminus Z) \subseteq \{1, 2, \ldots, r(d+1)\}$.
Moreover, we demand that $\pi$ does not leave any empty segment, i.e, 
for every $1 \leq i \leq r$, $\pi^{-1}(\setseg_i) \neq \emptyset$.

\begin{theorem}\label{thm:dist:fewseg}
  The extended \distname{} problem can be solved in
  $O(\stateconst^{|X \setminus Z|} n^{O(r)})$ time and space.
  If we are restricted to polynomial space,
  the extended \distname{} problem can be solved in $O(\pspaceconst^{|X \setminus Z|} n^{O(r \log n)})$ time.
\end{theorem}

\newcommand{\xn}{{\hat{n}}}

  Let $\xn = |X \setminus Z|$. The algorithm
  for Theorem \ref{thm:dist:fewseg} goes as follows.
  First, for each segment $1 \leq i \leq r$ we guess the vertex $v_i$ and position
  $1 \leq p_i \leq r(d+1)$ such that $\setseg(p_i) = i$.
  There are roughly $O(n^{O(r)})$ possible guesses (if $r > \xn$ the answer is immediately negative).
  We seek for embeddings $\pi$ such that for every $1 \leq i \leq r$
  position $\pi(v_i) = p_i$,
  and there is no vertex assigned to any position in the segment $i$ with color earlier than
  $\opcol(p_i)$, i.e., $v_i$ is the first vertex in the segment $i$.
  If there exists $z \in Z$ such that $\bar{\pi}(z) \leq 0$,
  then we require that $v_1$ is pushed by such $z$ that $\bar{\pi}(z)$ is the largest nonpositive possible.

Along the lines of the algorithm for \bwname{} \cite{naszicalp} and
algorithm for \distname{} by Fomin et al. \cite{fomin:distortion-5n},
we define state and a state successor as follows:

\begin{definition}
  A {\em{state}} is a triple $(p, (A, f), (H, h))$ such that:
  \begin{enumerate}
    \item $0 \leq p \leq r(d+1)$ is an integer,
    \item $(A, f)$ is a partial bucket function,
    \item $H \subseteq A$ is a set of vertices such that $H \cap \setseg_i$ is nonempty
      iff $f^{-1}(i)$ is nonempty,
    \item $h : H \to \setpos_p$ and if $v \in H$ then $f(v) = \opseg(h(v))$,
    \item if for any segment $1 \leq i \leq r$, vertex $v_i \in H$, then
      $h(v_i) = p_i$,
    \item if for any segment $1 \leq i \leq r$ position $p_i \in \setpos_i$ then $v_i \in A$ and
      $f(v_i) = i$.
  \end{enumerate}
\end{definition}

\begin{definition}
  We say that a state $(p+1, (A_2, f_2), (H_2, h_2))$ is a {\em{successor}} of a state
  $(p, (A_1, f_1), (H_1, h_1))$ iff:
  \begin{enumerate}
    \item $A_2 = A_1$ or $A_2 = A_1 \cup \{v\}$,
    \item if $A_2 = A_1$ then $f_2 = f_1$, $H_1 = H_2$ and $h_1 = h_2$,
    \item if $A_2 = A_1 \cup \{v\}$, then:
      \begin{enumerate}
        \item partial bucket function $(A_2, f_2)$ is a successor of the partial
          bucket function $(A_1, f_1)$ with the vertex $v$, such that
          $f_2(v) = \opseg(p+1)$,
        \item $H_2 = (H_1 \setminus f_1^{-1}(\opseg(p+1))) \cup \{v\}$,
        \item $h_2 = h_1|_{H_1 \cap H_2} \cup (v, p+1)$,
        \item if $H_1 \cap f_1^{-1}(\opseg(p+1)) = \{w\}$, then
          $d_G(v, w) = h_2(v) - h_1(w)$,
        \item for any $z \in Z$, $d_G(z, v) \leq |\bar{\pi}(z) - (p+1)| \leq d \cdot d_G(z, v)$.
      \end{enumerate}
  \end{enumerate}
\end{definition}

\begin{definition}
  We say that a state $(r(d+1), (V, f), (H, h))$ is a {\em{final state}} iff
  for each segment $1 \leq i \leq r$ we have $\{w_i\} = H \cap \setseg_i$ (i.e., $H \cap \setseg_i$ is nonempty),
  $w_i$ pushes $v_{i+1}$ for $i<r$ and $w_r$ pushes first $z \in Z$ such that $\bar{\pi}(z) \in \setseg_{r+1}$ (if such $z$ exists).
\end{definition}

The following equivalence holds:
\begin{lemma}\label{lem:dist:eq1}
  Let $\pi$ be a pushing embedding and a solution to the extended \distname{} problem with distortion
  at most $d$. Assume that $\pi(v_i) = p_i$ and this is the first vertex in the segment
  $i$ for every segment $1 \leq i \leq r$, i.e., the initial guesses are correct with
  respect to the solution $\pi$.
  For each $1 \leq p \leq r(d+1)$ we define $(A_p, f_p)$ and $(H_p, h_p)$ as follows:
  \begin{enumerate}
    \item $A_p = \pi^{-1}(\setpos_p)$ and $f_p = \opseg|_{A_p}$,
    \item for each segment $1 \leq i \leq r$ we take $w_i$ as the
      vertex in $\pi^{-1}(\setpos_p \cap \setseg_i)$ with the greatest color of
      position and take $w_i \in H_p$, $h_p(w_i) = \pi(w_i)$; 
      if $\pi^{-1}(\setpos_p \cap \setseg_i) = \emptyset$, we take $H_p \cap \setseg_i = \emptyset$.
  \end{enumerate}
  Then $S_p = (p, (A_p, f_p), (H_p, h_p)$ is a state and $S_{p+1} = (p+1, (A_{p+1}, f_{p+1}), (H_{p+1}, h_{p+1}))$
  is its successor if $p < r(d+1)$.
  Moreover, $S_{r(d+1)}$ is a final state.
\end{lemma}

\begin{proof}
  First note that, similarly as in the case of \bwname{}, $(A_p, f_p)$ is a partial
  bucket function and $(A_{p+1}, f_{p+1})$ is a successor of $(A_p, f_p)$. Indeed,
  the conditions for a partial bucket function and its successor are implied by Lemma \ref{lem:najwazniejsze-dist}.

  The check that $(H_p, h_p)$ satisfies the conditions for being a state is straightforward.
  Let us now look at the conditions for the successor. The only nontrivial part is that
  if in $H_p$ the vertex $w$ is replaced by $v$ in $H_{p+1}$, then $d_G(v, w) = h_{p+1}(v) - h_p(w)$.
  However, this is implied by the fact that $\pi$ is a pushing embedding.

  To see that $S_{r(d+1)}$ is a final state recall that $\pi$ leaves no segment $\setseg_i$, $1 \leq i \leq r$, nonempty
  and it is a pushing embedding.
\end{proof}

\begin{lemma}\label{lem:dist:eq2}
  Assume that we have a sequence of states $(S_p)_{p=0}^{r(d+1)}$, $S_p = (p, (A_p, f_p), (H_p, h_p))$
  such that $S_{p+1}$ is a successor of $S_p$ for $0 \leq p < r(d+1)$ and $S_{r(d+1)}$ is a final state.
  Let $\pi = \bigcup_{p=0}^{r(d+1)} h_p$. Then $\pi$ is a solution to the
  extended \distname{} problem with distortion at most $d$. Moreover, $\pi(v_i) = p_i$
  for all $1 \leq i \leq r$.
\end{lemma}

\begin{proof}
  Note that the conditions for the final state imply that $\pi$ leaves every segment
  from $1$ to $r$ nonempty. Moreover, the conditions for $(H_p, h_p)$ imply
  that $\pi(v_i) = p_i$ and $v_i$ is the first vertex assigned in segment $i$.

  First we check if $\pi$ is a pushing embedding. Let $v$ and $w$ be two vertices
  such that $\pi(v) < \pi(w)$ and there is no $u$ with $\pi(v) < \pi(u) < \pi(w)$.
  If $\opseg(v) = \opseg(w)$, then $\pi(w) - \pi(v) = d_G(v, w)$ is ensured by
  the state successor definition at step, where $S_{p+1}$ is a successor of the state $S_p$
  with the vertex $w$. Otherwise, if $\opseg(v) + 1 = \opseg(w)$, then $w = v_{\opseg(v)}$
  or $w$ is the first vertex of $Z$ in segment $r+1$ and the fact that $v$ pushes $w$
  is implied by the condition of the final state. The possibility that $\opseg(v) + 1 < \opseg(w)$
  is forbidden since in the final state $H_{r(d+1)} \cap \setseg_i \neq \emptyset$ for $1 \leq i \leq r$.

  Now we check if for each edge $uv$, $|\pi(u) - \pi(v)| \leq d$. Assume not,
  let $\pi(u) + d < \pi(v)$ and let $S_k$ be a successor of the state $S_{k-1}$ with the vertex $v$.
  By the conditions for a partial bucket function $(A_k, f_k)$, $|\opseg(u) - \opseg(v)| \leq 1$,
  so $\opseg(u) + 1 = \opseg(v)$. However, by the conditions for a partial bucket function
  successor, $\opcol(u) > \opcol(v)$, a contradiction, since consecutive positions of the same
  color are in distance $d+1$.
\end{proof}

Let us now limit the number of states.
There are at most $O^*(\stateconst^{\xn})$ partial bucket functions. Integer $p = O(rd)$ and
$h_p$ keeps position of at most one vertex in each segment, so there are $O(n^{O(r)})$
possible pairs $(H_p, h_p)$. Therefore, in total, we have $O(\stateconst^{\xn} n^{O(r)})$
states. Note that there at most $\xn+1$ successors of a given state, since
choosing $A_2 \setminus A_1$ defines the successor uniquely.
Note that, as checking if a pair $(A, f)$ is a partial bucket function can be done in polynomial time,
checking if a given triple is a state or checking if one state is a successor of the other
can be done in polynomial time too.

To obtain the $O(\stateconst^{\xn} n^{O(r)})$--time and space algorithm,
we simply seek a path of states as in Lemma \ref{lem:dist:eq2},
memoizing visited states. To limit the algorithm to the polynomial space,
we do the same trick as in the $O(\pspaceconst^n)$ algorithm for \bwname{}.

\begin{lemma}\label{lem:dist:stateguess}
  Assume that we have states $S_p = (p, (A_p, f_p), (H_p, h_p))$ and
  $S_q = (q, (A_q, f_q), (H_q, h_q))$ such that
  $p < q$, $A_p \subseteq A_q$ and $f_p = f_q|_{A_p}$.
  Let $m = |A_q \setminus A_q|$. Then one can check
  if there exists a sequence of states $S_i = (i, (A_i, f_i), (H_i, h_i))$ for $i = p, p+1, \ldots, q$
  such that the state $S_{i+1}$ is a successor of the state numbered $S_i$
  in time $O(4^m n^{O(r \log m)})$.
\end{lemma}

\begin{proof}
  First, let us consider the case when $m = 1$. We guess index $k$, $p < k \leq q$,
  such that $A_k = A_q$ and $f_k = f_q$, but $A_{k-1} = A_p$ and $f_{k-1} = f_p$.
  Note that then all states $S_i$ for $p \leq i \leq q$ are defined uniquely: $h_i = h_p$
  for $i < k$ and $h_i = h_q$ for $i \geq k$. We need only to check if all consecutive pairs
  of states are successors.

  Let now assume $m > 1$ and let $s = |A_p| + \lfloor m/2 \rfloor$.
  Let us guess the state $S_k$ such that $|A_k| = s$. We need
  $A_p \subseteq A_k \subseteq A_q$ and $f_k = f_q|_{A_k}$, so we have
  only roughly $2^m$ possibilities for $(A_k, f_k)$ and $O(dr) = O(n\xn)$ possibilities
  for the index $k$. As always, there are $n^{O(r)}$ possible guesses for $(H_k, h_k)$.
  We recursively check if there is a sequence of states from $S_p$ to $S_k$
  and from $S_k$ to $S_q$. Since at each step we divide $m$ by $2$, finally we obtain
  an $O(4^m n^{O(r \log m)})$ time bound.
\end{proof}

Again we set $\alpha := \alphaprog$.

\begin{enumerate}
\item We guess the state $S_k = (k, (A_k, f_k), (H_k, h_k))$
  such that $|A_k| = \lfloor \alpha n \rfloor $.
  By Theorem \ref{thm:20n2} with $N = n$, we can enumerate all partial bucket extensions
in $O(\stateconst^{\xn})$. We enumerate all partial bucket functions,
guess $p$ and $(H_k, h_k)$ and drop those combinations that are not states.
Note that there are $O(n^{O(r)})$ possible guesses for $(H_k, h_k)$ and $dr \leq n^2$ guesses
for $p$.
\item Using Lemma \ref{lem:dist:stateguess}, check if there is a path of states from $(0, (\emptyset, \emptyset), (\emptyset, \emptyset))$ to $S_k$.
  This phase works in time $4^{\alpha \xn} n^{O(r \log n)}$.
  In total, for all $(A_k, f_k)$, this phase works in time
  $O^*(\stateconst^{\xn} \cdot 4^{\alpha \xn} n^{O(r \log n)}) = O(\pspaceconst^\xn n^{O(r \log n)})$.
\item Guess the final state
  $S_{r(d+1)} = (r(d+1), (V, f_{r(d+1)}), (H_{r(d+1)}, h_{r(d+1)}))$: $f_{r(d+1)}$
  needs to be a bucket extension of the partial bucket function $(A_k, f_k)$.
By Lemma \ref{lem:extpoly}, bucket extensions can be enumerated with polynomial delay.
We guess $h_{r(d+1)}$ and simply drop those guesses that do not form states.
By Lemma \ref{lem:5n} with $N = r$, there will be at most $O^*(5^\xn)$ pairs of states $(A_k, f_k)$ and $(V, f_{r(d+1)})$. We have $n^{O(r)}$ possibilities for $h_{r(d+1)}$.
\item Using Lemma \ref{lem:dist:stateguess},
  check if there is a path from the state $S_k$ to $S_{r(d+1)}$.
  This phase works in time $4^{(1-\alpha)n} n^{O(r \log n)}$.
  In total, for all $S_k$ and $S_{r(d+1)}$ this phase
  works in time 
  $$O^*(5^\xn 4^{(1-\alpha)\xn} n^{O(r \log n)}) = O(\pspaceconst^\xn n^{O(r \log n)}).$$
\end{enumerate}

\begin{theorem}
  The \distname{} problem can be solved in $O(\stateconst^n)$ time and space.
  If we are restricted to polynomial space,
  the extended \distname{} problem can be solved in $O(\pspaceconst^n)$ time.
\end{theorem}

\begin{proof}
  We almost repeat the argument from \cite{fomin:distortion-5n}.
First, we may guess the number of nonempty segments needed to embed $G$
into a line with a pushing embedding $\pi$
  with distortion at most $d$. Denote this number by $r$, i.e., $r = \lceil \max \{\pi(v) : v \in V(G)\} / (d+1) \rceil$.
Note that the original \distname{} problem can be represented as an extended case with $H = G$ and $Z = \bar{\pi} = \emptyset$ and
with guessed $r$.

If $r < n / \log^3(n)$, the thesis is straightforward by applying Theorem \ref{thm:dist:fewseg}. Therefore, let us assume
$r \geq n / \log^3(n)$. As every segment from $1$ to $r$ contains at least one vertex in a required pushing embedding $\pi$,
by simple counting argument, there needs to be a segment $r/4 \leq k \leq 3r/4$ such that there are at most $4n/r \leq 4\log^3(n)$ vertices 
assigned to segments $k$ and $k+1$ in total by $\pi$.
We guess: segment number $k$, vertices assigned to segments $k$ and $k+1$
and values of $\pi$ for these vertices. We discard any guess that already
makes some edge between guessed vertices longer than $d$.
As $d, r = O(n)$, we have $n^{O(\log^3 n)}$ possible guesses.

Let $Y$ be the set of vertices assigned to segments $k$ and $k+1$ and look at any 
connected component $C$ of $G[V \setminus Y]$. Note that if $C$ has neighbours in
both segment $k$ and $k+1$, the answer is immediately negative. Moreover, as $G$ was connected,
$C$ has a neighbour in segment $k$ or $k+1$. Therefore we know, whether vertices from $C$
should be assigned to segments $1, 2, \ldots, k-1$ or $k+2, \ldots, r$.
The problem now decomposes into two subproblems: graphs $H_1$ and $H_2$, such
that $H_1$ should be embedded into segments $1$ to $k$ and $H_2$ should be embedded
into segments $k+1$ to $r$; moreover, we demand that the embeddings meet the guesses values
of $\pi$ on $Y$.

The subproblems are in fact instances of extended \distname{} problem and can be decomposed further
in the same fashion until there are at most $n / \log^3(n)$ segments in one instance.
The depth of this recurrence is $O(\log r) = O(\log n)$, and each subproblem with
at most $n / \log^3(n)$ can be solved by algorithm described in Theorem \ref{thm:dist:fewseg}.
Therefore, finally, we obtain an algorithm that works in $O(\stateconst^n)$ time and space
and an algorithm that works in $O(\pspaceconst^n)$ time and polynomial space.
\end{proof}

\bibliographystyle{plain}
\bibliography{bandwidth-distortion-revisited}

\newpage

\appendix

\section{Bound on the number of partial bucket functions}\label{a:20n2}

In this section we prove Theorem \ref{thm:20n2}; namely, that for some constant $c < \stateconst$
in a connected, undirected graph $G = (V, E)$ with $|V| = n$ there
are at most $O(N \cdot c^n)$ bucket functions, where we are allowed
to assign values $\{1, 2, \ldots, N\}$ only.
Let $c = \stateconst - \varepsilon$
for some sufficiently small $\varepsilon$.
We use $c$ instead of simply constant $\stateconst$
to hide polynomial factors at the end, i.e., to say $O^*(c^n) = O(\stateconst^n)$.

Let us start with the following observation.

\begin{lemma}
  Let $G' = (V, E')$ be a graph formed by removing one edge from the graph $G$
  in a way that $G'$ is still connected. If $(A, f)$ is a bucket function in $G$,
  then it is also a bucket function in $G'$.
\end{lemma}

Therefore we can assume that $G = (V, E)$ is a tree. Take any vertex
$v_r$ with degree $1$
and make it a root of $G$. 


In this proof we limit not the number of partial bucket functions, but the number of
{\it{prototypes}}, defined below. It is quite clear that the number
of prototypes is larger than the number of partial bucket extensions, and we prove
that there are at most $O(Nc^n)$ prototypes.
Then we show that one can generate all prototypes in $O^*(Nc^n)$ time
and in polynomial space. This proves that all partial bucket extensions
can be generated in $O^*(Nc^n)$ time and polynomial space.

\begin{definition}\label{def:genex-prestate}
  Assume we have a fixed subset $B \subseteq V$.
  A {\it{prototype}} is a pair $(A, f)$, where $A \subseteq V$,
  $f : A \cup B \to \Z$, such that $(A, f|_A)$ is a partial bucket function,
  and there exists a bucket extension $\extf$ that is an extension of $f$,
  not only $f|_A$.
\end{definition}

\begin{lemma}
  For any fixed $B \subseteq V$ the number of partial
  bucket functions in not greater than the number of prototypes.
\end{lemma}

\begin{proof}
  Let us assign to every prototype $(A, f)$ the partial bucket function $(A, f|_A)$.
  To prove our lemma we need to show that this assignment is surjective.
  Having a partial bucket function $(A, f)$,
  take any its bucket extension $\extf$ and
  look at the pair $(A, \extf|_{A \cup B})$.
  This is clearly a prototype, and $(A, f)$ is assigned to it in the aforementioned
  assignment.
\end{proof}

Before we proceed to main estimations, we need a few calculations.
Let $\alpha = 4.26$, $\beta = 3$ and $\gamma = 5.02$.

\begin{lemma}\label{lem:suma357}
  \begin{equation*}
    2c^{n-1} + \sum_{k=1}^\infty (2k-1)c^{n-k} = c^n \Big(\frac{2}{c} + \frac{2c}{(c-1)^2} - \frac{1}{c-1} \Big)
  \end{equation*}
\end{lemma}

\begin{proof}
  \begin{equation}
    \sum_{k=1}^\infty kc^{-k} = \frac{1}{c} \sum_{k=0}^\infty (k+1)c^{-k} = \frac{1}{c} \Big( \frac{1}{1-x} \Big)' \Big|_{x=\frac{1}{c}} = \frac{c}{(c- 1)^2} \label{r:sumakck}
  \end{equation}
  \begin{align*}
    2c^{n-1} + \sum_{k=1}^\infty (2k-1)c^{n-k} = \\
    = c^n \Big( 2\sum_{k=1}^\infty kc^{-k} - \sum_{k=1}^\infty c^{-k} + 2c^{-1} \Big) = \\
    = c^n \Big( \frac{2}{c} + \frac{2c}{(c- 1)^2} - \frac{1}{c - 1}\Big)
  \end{align*}
\end{proof}

\begin{corollary}\label{cor:suma357}
  For our choice of values for $\alpha$, $\gamma$ and $c$ we obtain
  \begin{equation*}
    2c^{n-1} + \sum_{k=1}^\infty (2k-1)c^{n-k} \leq c^n \left( 1 - \max\left(\frac{6}{\alpha c^2}, \frac{15}{\gamma c^3}\right) \right).
  \end{equation*}
\end{corollary}

\begin{lemma}\label{lem:suma246}
  \begin{equation*}
    \sum_{k=1}^\infty 2kc^{n-k} = c^n \frac{2c}{(c-1)^2}
  \end{equation*}
\end{lemma}

\begin{proof}
  This is a straightforward corollary from Equation \ref{r:sumakck}.
\end{proof}

\begin{corollary}\label{cor:suma246}
  For our choice of values for $\beta$, $\gamma$ and $c$ we obtain
  \begin{equation*}
    \sum_{k=1}^\infty 2kc^{n-k} \leq c^n \left(1 - \max\left(\frac{7}{\beta c^2}, \frac{13}{\gamma c^2}\right)\right).
  \end{equation*}
\end{corollary}

Let us proceed to the main estimations.

\begin{lemma}\label{lem:Tn}
  Let $G$ be a path of length $n+1$ --- graph with $V = \{v_0, v_1, v_2, \ldots, v_n\}$,
  $E = \{(v_i, v_{i+1}): 0 \leq i < n\}$. Let $B = \{v_0\}$.
  Fix any $j \in \Z$. Let $T(n)$ be the number of prototypes $(A, f)$ satisfying
  $v_0 \in A$ and $f(v_0) = j$. Then $T(n) \leq \alpha \cdot c^{n-1}$.
\end{lemma}

\begin{proof}
  Let us denote $T(x) = 0$ for $x \leq 0$.
  This satisfies $T(x) \leq \alpha c^{x-1}$.
  We use the induction and start with calculating $T(1)$ and $T(2)$ manually.

  If $n=1$ we have $f(v_1) \in \{j-1, j, j+1\}$ if $v_1 \in A$, and one prototype if $v_1 \notin A$,
  so $T(1) = 4 < \alpha$.

  If $n=2$, we consider several cases.
  If $v_1 \in A$ we have $f(v_1) \in \{j-1, j, j+1\}$ and $T(1)$
  possibilities for $A \setminus \{v_0\}$ and $f|_{A \setminus \{v_0\}}$.
  If $A = \{v_0, v_2\}$, $f(v_2) \in \{j-1, j, j+1\}$ due to the conditions
  for a partial bucket extension $\bar{f}$.
  There is also one state with $A = \{v_0\}$, ending up with
  $T(2) = 3 \cdot 4 + 3 + 1 = 16 < \alpha c$.

  Let us recursively count interesting prototypes for $n \geq 3$.
  There is exactly one prototype $(A, f)$ with $A = \{v_0\}$. Otherwise let
  $k(A) > 0$ be the smallest positive integer satisfying $v_{k(A)} \in A$. Let us count
  the number of prototypes $(A, f)$, such that $k(A) = k$ for fixed $k$.

  For $k=1$ we have $f(v_1) \in \{j-1, j, j+1\}$, and, having
  fixed value $f(v_1)$, we have $T(n-1)$ ways to choose
  $A \setminus \{v_0\}$ and $f_{A \setminus \{v_0\}}$.

  For $k>1$ we have $j - k + 1 \leq f(v_k) \leq j + k - 1$, due to the conditions
  for a partial bucket extension $\bar{f}$, so we
  have $(2k-1)$ ways to choose $f(v_k)$
  and $T(n-k)$ ways to choose
  $A \setminus \{v_0, v_1, \ldots, v_{k-1}\}$ and $f_{A \setminus \{v_0, v_1, \ldots, v_{k-1}\}}$
  if $k < n$ and $1$ way if $k = n$.

  Therefore we have for $n \geq 3$:
  \begin{align*}
    T(n) \leq 1 + 3T(n-1) + \sum_{k=2}^{n-1} (2k-1) T(n-k) + 2n-1 \leq\\
    \leq 2n + 2T(n-1) + \sum_{k=1}^{\infty} (2k-1)T(n-k)
  \end{align*}
  Note that for $n \geq 3$ we have $2n \leq \frac{6}{\alpha c^2} \cdot \alpha c^{n-1}$, as we have an equality for $n=3$ and
  the right side grows significantly faster than the left side for $n \geq 3$. Using
  Corollary \ref{cor:suma357} we obtain:
  \begin{equation*}
    T(n) \leq \alpha c^{n-1}
  \end{equation*}
\end{proof}

\begin{lemma}\label{lem:Tpn}
  Let $G$ be a path of length $n+1$ --- graph with $V = \{v_0, v_1, v_2, \ldots, v_n\}$,
  $B = \{v_0\}$ and
  $E = \{(v_i, v_{i+1}): 0 \leq i < n\}$.
  Fix any $j \in \Z$. Let $T'(n)$ be the number of prototypes $(A, f)$ satisfying
  $v_0 \notin A$ and $f(v_0) = j$. Then $T'(n) \leq \beta c^{n-1}$.
\end{lemma}

\begin{proof}
  Write the formula for $T'$ using previously bounded $T$.
  We start with calculating $T'(1)$ and $T'(2)$ manually.

  If $n=1$,
  if $v_1 \in A$ we have $f(v_1) \in \{j, j+1\}$ and one prototype with $A = \emptyset$,
  so $T'(1) = 3 \le \beta $.

  If $n=2$, we have one prototype with $A = \emptyset$,
  four prototypes if $A = \{v_2\}$ (since then $f(v_2) \in \{j-1, j, j+1, j+2\}$)
  and $2 \cdot T(1)$ prototypes if $v_1 \in A$ (since $f(v_1) \in \{j, j+1\}$).
  Therefore $T'(2) = 1 + 4 + 2 \cdot 4 = 13 < \beta c$.

  Let us assume $n \geq 3$.

  There is exactly one prototype $(A, f)$ with $A = \emptyset$. Otherwise
  let $k(A) > 0$ be the smallest positive integer satisfying $v_{k(A)} \in A$. Let us count the number
  of prototypes $(A, f)$ such that $k(A) = k$ for fixed $k$.

  Note that, due to the conditions for a partial bucket extension $\bar{f}$,
  $j - k + 1 \leq f(v_k) \leq j + k$; there are $2k$ ways to choose
  $f(v_k)$.
  There are $T(n-k)$ ways to choose
  $A \setminus \{v_0, v_1, \ldots, v_{k-1}\}$ and $f_{A \setminus \{v_0, v_1, \ldots, v_{k-1}\}}$
  for $k < n$ and $1$ way for $k = n$, leading us to inequality

  \begin{equation*}
    T'(n) \leq 1 + 2n + \sum_{k=1}^\infty 2kT(n-k)
  \end{equation*}
  Note that for $n \geq 3$ we have $2n + 1 \leq \frac{7}{\beta c^2} \cdot \beta c^{n-1}$,
  as we have equality for $n=3$ and the right side grows significantly faster than the left side for $n \geq 3$.
  Therefore, using Corollary \ref{cor:suma246}, we obtain
  \begin{equation*}
    T'(n) \leq \beta c^{n-1}
  \end{equation*}
\end{proof}

\begin{lemma}\label{lem:Sn}
  Let $G$ be a path of length $n+1$ --- graph with $V = \{v_0, v_1, v_2, \ldots, v_n\}$,
  $B = \{v_0, v_n\}$ and
  $E = \{(v_i, v_{i+1}): 0 \leq i < n\}$.
  Fix any $j \in \Z$. Let $S(n)$ be the number of prototypes $(A, f)$ satisfying
  $v_0 \in A$ and $f(v_0) = j$. Then $S(n) \leq \gamma c^{n-1}$.
  Moreover, at least $0.4 S(n)$ of these prototypes $(A, f)$ satisfy $v_n \notin A$.
\end{lemma}

\begin{proof}
  As in the estimations of $T(n)$, we use induction and write a recursive formula for $S$.
  Let $S(x) = 0$ for $x \leq 0$.

  We start with calculating $S(1)$, $S(2)$ and $S(3)$ manually.
  If $n=1$, if $v_1 \in A$ we have $f(v_1) \in \{j-1, j, j+1\}$ and if $v_1 \notin A$
  we have $f(v_1) \in \{j-1, j\}$, thus $S(1) = 5 \leq \gamma$ and
  $2 = 0.4S(1)$ of these prototypes satisfy $v_1 \notin A$.

  If $n=2$, we consider several cases, as in calculations of $T(2)$.
  If $v_1 \in A$, we have $f(v_1) \in \{j-1, j, j+1\}$ thus $3 \cdot S(1)$ possibilities and out of them $3 \cdot 2$ possibilities satisfy $v_2 \notin A$.
  If $A = \{v_0, v_2\}$ we have $f(v_2) \in \{j-1, j, j+1\}$, $3$ possibilities.
  If $A = \{v_0\}$ we have $f(v_2) \in \{j-2, j-1, j, j+1\}$, $4$ possibilities.
  In total, $S(2) = 15 + 3 + 4 = 22 \leq \gamma c$ and $3 \cdot 2 + 4 > 0.4S(2)$ of these prototypes
  satisfy $v_2 \notin A$.

  If $n=3$, we do similarly.
  If $v_1 \in A$, we have $f(v_1) \in \{j-1, j, j+1\}$ thus $3 \cdot S(2)$ possibilities and out of them $3 \cdot 10$ possibilities satisfy $v_3 \notin A$.
  If $v_1 \notin A$ but $v_2 \in A$ we have $f(v_2) \in \{j-1, j, j+1\}$, $3 \cdot S(1)$ possibilities and out of them $3 \cdot 2$ possibilities satisfy $v_3 \notin A$.
  If $A = \{v_0, v_3\}$ we have $f(v_3) \in \{j-2, j-1, j, j+1, j+2\}$, $5$ possibilities.
  If $A = \{v_0\}$ we have $f(v_3) \in \{j-3, j-2, j-1, j, j+1, j+2\}$, $6$ possibilities.
  In total $S(3) = 3 \cdot 22 + 3 \cdot 5 + 5 + 6 = 92 \leq \gamma c^2$, and $3 \cdot 10 + 3 \cdot 2 + 6 = 42 > 0.4S(3)$ of these prototypes
  satisfy $v_3 \notin A$.

  Let us assume $n \geq 4$.
  If $A = \{v_0\}$, we have $j - n \leq f(v_n) \leq j + n - 1$, $2n$ possible
  prototypes and all of them satisfy $v_n \notin A$. Otherwise
  let $k(A)$ be the smallest positive integer such that $v_{k(A)} \in A$. Let us once again count the number
  of prototypes $(A, f)$, such that $k(A) = k$ for fixed $k$.

  As in the estimate of $T(n)$, we have $3$ possible values for
  $f(v_k)$ when $k=1$ and $(2k-1)$ possible values when $k > 1$.
  For $k < n$ there are $S(n-k)$ possible ways to choose
  $A \setminus \{v_0, v_1, \ldots, v_{k-1}\}$ and $f_{A \setminus \{v_0, v_1, \ldots, v_{k-1}\}}$
  and $1$ way if $k = n$.
  Moreover for $k < n$ at least $0.4S(n-k)$ of choices satisfy $v_n \notin A$.
  Therefore:
  \begin{equation*}
    S(n) = 2n-1 + 2n + 2S(n-1) + \sum_{k=1}^{n-1}(2k-1)S(n-k)
  \end{equation*}
  And at least 
  \begin{equation*}
     2n + 0.4\left(2S(n-1) + \sum_{k=1}^{n-1}(2k-1)S(n-k)\right) \geq 0.4S(n)
  \end{equation*}
   of these
  prototypes satisfy $v_n \notin A$.
  For $n \geq 4$ we have $4n-1 \leq \frac{15}{\gamma c^3} \cdot \gamma c^{n-1}$, so using
  Corollary \ref{cor:suma357} we obtain:
  \begin{equation*}
    S(n) \leq \gamma c^{n-1}
  \end{equation*}
\end{proof}

\begin{lemma}\label{lem:Spn}
  Let $G$ be a path of length $n+1$ --- graph with $V = \{v_0, v_1, v_2, \ldots, v_n\}$,
  $B = \{v_0, v_n\}$ and
  $E = \{(v_i, v_{i+1}): 0 \leq i < n\}$.
  Fix any $j \in \Z$. Let $S'(n)$ be the number of prototypes $(A, f)$ satisfying
  $v_0 \notin A$ and $f(v_0) = j$. Then $S'(n) \leq \gamma c^{n-1}$.
  Moreover, at least $0.4 S'(n)$ of these prototypes $(A, f)$ satisfy $v_n \notin A$.
\end{lemma}

\begin{proof}
  Similarly to the estimate of $T'$, we write the formula bounding $S'$ with $S$ and
  use already proved bounds for $S$.
  We start with calculating $S'(1)$ and $S'(2)$ manually.

  If $n=1$ we have $f(v_1) \in \{j, j+1\}$ if $v_1 \in A$ and
  $f(v_1) \in \{j-1, j, j+1\}$ if $v_1 \notin A$, thus $S'(1) = 5 \leq \gamma$
  and $3 > 0.4S'(1)$ of these prototypes satisfy $v_1 \notin A$.

  If $n=2$ we consider several cases.
  If $v_1 \in A$ we have $f(v_1) \in \{j, j+1\}$, thus $2 \cdot S(1)$ possibilities and out of them $2 \cdot 2$ possibilities satisfy $v_2 \notin A$.
  If $A = \{v_2\}$ we have $f(v_2) \in \{j-1, j, j+1, j+2\}$, $4$ possibilities.
  If $A = \emptyset$ we have $f(v_2) \in \{j-2, j-1, j, j+1, j+2\}$, $5$ possibilities.
  In total $S'(2) = 2 \cdot 5 + 4 + 5 = 19 \leq \gamma c$, and $2 \cdot 2 + 5 = 9 > 0.4'S(2)$ of these prototypes satisfy $v_2 \notin A$.

  Let us assume $n \geq 3$.
  If $A = \emptyset$, we have $j - n \leq f(v_n) \leq j + n$, $2n+1$ possible
  prototypes, all satisfying $v_n \notin A$. Otherwise
  let $k(A)$ be the smallest positive integer such that $v_{k(A)} \in A$. Let us once again count number
  of prototypes $(A, f)$, such that $k(A) = k$ for fixed $k$.

  As in the estimate of $T'(n)$, we have $2k$ possible values for
  $f(v_k)$.
  For $k < n$ there are $S(n-k)$ possible ways to choose
  $A \setminus \{v_0, v_1, \ldots, v_{k-1}\}$ and $f_{A \setminus \{v_0, v_1, \ldots, v_{k-1}\}}$
  and $1$ way if $k = n$.
  Moreover, for $k<n$ at least $0.4S(n-k)$ of choices satisfy $v_n \notin A$. Therefore:
  \begin{equation*}
    S'(n) \leq 2n+1 + 2n + \sum_{k=1}^\infty 2kS(n-k)
  \end{equation*}
  and at least 
  \begin{equation*}
    2n+1 + 0.4\sum_{k=1}^\infty 2kS(n-k)\geq 0.4 S'(n)
  \end{equation*}
  of these prototypes
  satisfy $v_n \notin A$.
  For $n \geq 3$ we have $4n+1 \leq \frac{13}{\gamma c^2} \cdot \gamma c^{n-1}$. Using Corollary \ref{cor:suma246}
  we obtain
  \begin{equation*}
    S'(n) \leq \gamma c^{n - 1}
  \end{equation*}
\end{proof}

Let us proceed to the final lemma in this proof.
By $B_0 \subseteq V$ we
denote the root $v_r$ and the set of vertices with at least two children in $G$, i.e., vertices
of degree at least $3$. Recall that $v_r$ has degree $1$.

\begin{lemma}
  Let $v_r$ be the root of an $n$ vertex graph $G=(V, E)$ of degree $1$ and let
  $B = B_0$. Assume that $G$ is not a path. Fix $j \in \Z$. Then both the number of prototypes
  $(A, f)$ with $f(v_r) = j$, $v_r \in A$ and
  the number of prototypes $(A, f)$ with $f(v_r) = j$, $v_r \notin A$ are at most $\delta c^{n-2}$,
  where $\delta = \sqrt{0.6 \alpha^2 + 0.4 \beta^2}$.
\end{lemma}

\begin{proof}
  We prove it by induction over $n = |V|$. Let $v$ be the closest
  to $v_r$ vertex that belongs to $B_0$ different than $v_r$ ($v$ exists as $G$ is not a path)
  Let $P$ be the path from $v$ to $v_r$, including $v$ and $v_r$ and
  let $|P|$ be the number of vertices on $P$. Due to Lemma \ref{lem:Sn}
  and Lemma \ref{lem:Spn}, there are at most $\gamma c^{|P|-2}$ ways to choose
  $(A \cap P, f|_{(A \cup B) \cap P})$, and at least $0.4$ of these possibilities
  satisfy $v \notin A$. Let us now fix one of such choices.

  Let $G_1$, $G_2$, \ldots, $G_k$
  be the connected components of $G$ with removed $P$. Let $V_i$ be the set of vertices
  of $G_i$ and $B_i = B \cap V_i$. For each $1 \leq i \leq k$, we bound the number
  of possible choices for $(A \cap V_i, f|_{(A \cup B) \cap V_i})$.

  If $B_i = \emptyset$ (equivalently $G_i$ is a path) then
  one can choose $(A \cap V_i, f|_{(A \cup B) \cap V_i})$ on $T(|V_i|) \leq \alpha c^{|V_i|-1}$
  or $T'(|V_i|) \leq \beta c^{|V_i|-1}$ ways, depending on whether
  $v= v_0 \in A$ or $v = v_0 \notin A$ (we use here Lemma \ref{lem:Tn} or
  Lemma \ref{lem:Tpn} for $v_0 = v$ and $\{v_1, v_2, \ldots, v_{|V_i|}\} = V_i$).

  Otherwise, we use inductive assumption for $G_i$ with added root $v$.
  In this case we have at most $\delta c^{|V_i|-1}$ possibilities to choose $(A \cap V_i, f|_{(A \cup B) \cap V_i})$.

  \newcommand{\Beta}{\mathcal{B}}
  \newcommand{\Alpha}{\mathcal{A}}

  Let $\Beta = \{1 \leq i \leq k: B_i = \emptyset\}$, and $\Alpha = \{1, 2, \ldots, k\} \setminus \Beta$.
  If $v \in A$, the number of choices for all graphs $G_i$ is bounded by:
  \begin{equation*}
    \left( \prod_{i \in \Alpha} \delta c^{|V_i|-1} \right) \cdot \left(\prod_{i \in \Beta} \alpha c^{|V_i| - 1}\right) = \delta^{|\Alpha|}\alpha^{|\Beta|} c^{n-|P|-k}
  \end{equation*}
  If $v \notin A$, the number of choices for all graphs $G_i$ is bounded by:
  \begin{equation*}
    \left( \prod_{i \in \Alpha} \delta c^{|V_i|-1} \right) \cdot \left(\prod_{i \in \Beta} \beta c^{|V_i| - 1}\right) = \delta^{|\Alpha|}\beta^{|\Beta|} c^{n-|P|-k}
  \end{equation*}
  Therefore, as $\alpha \geq \beta$, the total number of prototypes for $G$ is bounded by
  \begin{equation*}
    \gamma c^{|P|-2} \delta^{|\Alpha|} c^{n-|P|-k} \left(0.6 \alpha^{|\Beta|} + 0.4 \beta^{|\Beta|}\right) = c^{n-2} \left(\gamma c^{-k} \delta^{|\Alpha|} \left(0.6 \alpha^{|\Beta|} + 0.4 \beta^{|\Beta|}\right)\right)
  \end{equation*}
  Note that $\delta \gamma \leq c^2$. If $\Beta \leq 1$ we have, using that $k \geq 2$ and $0.6\alpha +0.4\beta \leq \delta \leq c$:
  \begin{equation*}
    \gamma c^{-k} \delta^{|\Alpha|} \left(0.6 \alpha^{|\Beta|} + 0.4 \beta^{|\Beta|}\right) \leq \gamma c^{-k} \delta^{k} \leq \delta.
  \end{equation*}
  Otherwise, if $|\Beta| \geq 2$ we have, as $\beta \leq \alpha \leq c$ and $\delta \leq c$:
  \begin{align*}
    \gamma c^{-k} \delta^{|\Alpha|} \left(0.6 \alpha^{|\Beta|} + 0.4 \beta^{|\Beta|}\right) &\leq 
       \gamma c^{-k} \delta^{|\Alpha|} \left(0.6 \alpha^{|\Beta|} + 0.4 \alpha^{|\Beta|-2} \beta^2\right) \\
     &= \gamma c^{-k}\delta^{|\Alpha|} \alpha^{|\Beta|-2} \delta^2 \leq \delta.
  \end{align*}
  Thus the bound is proven.
\end{proof}

\begin{corollary}
  The number of all prototypes satisfying $f(v_r) \in \{1,2, \ldots, N\}$
  is at most $N \cdot \max(\alpha, \delta) \cdot c^{n-2} = O(Nc^n)$.
\end{corollary}

To finish up the proof of theorem \ref{thm:20n2}, we need to show the following lemma.

\begin{lemma}
  Fix $B = B_0$.
  All prototypes can be generated in polynomial space and in $O^*(Nc^n)$ time.
\end{lemma}

\begin{proof}
  We assume that $G=(V,E)$ is a tree rooted at $v_r$. Otherwise, we may take any spanning
  tree of $G$, generate all prototypes for this tree, and finally for each prototype
  in the spanning tree check if this is a prototype in the original graph $G$ too.

  First we guess $f(v_r)$ and guess the set $A$. Then we go in the root--to--leaves
  order in $G$ and guess values of $f$ for vertices in $A \cup B$. Whenever we encounter
  a vertex $v \in A \cup B$ we look at its closest predecessor $w \in A \cup B$.
  Let $d$ be the distance between $v$ and $w$. We iterate over all possibilities
  $f(v) \in \{f(w) -d, f(w)-d+1, \ldots, f(w)+d\}$; however the following options are forbidden
  due to the conditions for the bucket extension:
  \begin{itemize}
    \item if $v \in A$, $w \in A$ and $d > 1$ then $f(v) = f(w)-d$ and $f(v)=f(w)+d$ are forbidden;
    \item if $v \in A$ and $w \notin A$ then $f(v) = f(w) -d$ is forbidden;
    \item if $v \notin A$ and $w \in A$ then $f(v) = f(w)+d$ is forbidden.
  \end{itemize}
  Since every branch in our search ends up with a valid prototype, the algorithm
  takes $O^*(Nc^n)$ time. In memory, we keep only the recursion stack
  of the search algorithm, and therefore we use polynomial space.
\end{proof}

\section{Omitted proofs}\label{a:proofs}

\begin{proof}[Proof of Lemma \ref{lem:extpoly}]
  We construct all valid bucket extensions by a brute--force search. We start with $f'=f$ and $B=A$.
  At one step we have $A \subseteq B \subseteq V$, $f':B \to V$ such that $f'|_A = f$ and there exists
  a bucket extension $\extf$ of $(A, f)$ such that $\extf|_B = f'$. We take any $v \in V \setminus B$
  such that there exists a neighbour $w$ of $v$ that belongs to $B$ and try to assign $f'(v) = f'(w) + \varepsilon$,
  for each $\varepsilon \in \{-1, 0, 1\}$. At every step, we use the algorithm from Lemma \ref{lem:checkext} to check
  the condition if $f'$ can be extended to a valid bucket extension of $(A, f)$. This check ensures that
  every branch in our search algorithm ends up with a bucket extension. Therefore we generate all bucket
  extensions with a polynomial delay and in polynomial space.
\end{proof}

\begin{proof}[Proof of Lemma \ref{lem:najwazniejsze-dist}]
  First, assume $\pi$ has distortion at most $d$.
  Then for each $uv \in E$ we have $|\pi(u) - \pi(v)| \leq d$. Since
  segments are of size $d+1$, this implies that $|\opseg(u) - \opseg(v)| \leq 1$.
  Moreover, the distance between positions of the same color in consecutive segments
  is exactly $d+1$, which implies that if $\opseg(u) + 1 = \opseg(v)$ then $\opcol(u) > \opcol(v)$.

  In the other direction, assume that for some $u, v \in V$ we have $k = d_G(u, v)$
  $|\pi(u) - \pi(v)| > dk$. Let $u = u_0, u_1, \ldots u_k = v$ be the path
  of length $k$ between $u$ and $v$. Then, for some $0 \leq i < k$ we have
  $|\pi(u_{i+1}) - \pi(u_i)| > d$. This implies that $\opseg(u_{i+1}) \neq \opseg(u_i)$,
  w.l.o.g. assume that $\opseg(u_i) + 1 = \opseg(u_{i+1})$. However, since
  consecutive positions of the same color are in distance $d+1$, this implies that
  $\opcol(u_i) \leq \opcol(u_{i+1})$, a contradiction.
\end{proof}

\end{document}